\documentclass[journal]{IEEEtran}
\ifCLASSINFOpdf
\else
\fi
\usepackage{authblk}

\usepackage{graphicx}
\usepackage{braket}
\usepackage{amsthm}
\usepackage{amsmath}
\usepackage{amssymb}
\usepackage{multirow}
\newtheorem{theorem}{Theorem}[section]

\newtheorem{lemma}[theorem]{Lemma}

\makeatletter
\def\footnoterule{\kern-3\p@
  \hrule \@width 2in \kern 2.6\p@} 
\makeatother

\begin{document}
%

\title{Optimal Error Correcting Code For\\ Ternary Quantum Systems}
\author[1,*]{Ritajit Majumdar}
\author[1]{Susmita Sur-Kolay}

\affil[1]{Advanced Computing \& Microelectronics Unit, Indian Statistical Institute, India}
\affil[*]{Email: ritajit.majumdar@ieee.org}

%



\maketitle
\thispagestyle{empty}
\begin{abstract}
Multi-valued quantum systems can store more information than binary ones for a given number of quantum states. For reliable operation of multi-valued quantum systems, error correction is mandated. In this paper, we propose a 5-qutrit quantum error correcting code and provide its stabilizer formulation. Since 5 qutrits are necessary to correct a single error, our proposed code is optimal in the number of qutrits. We prove that the error model considered in this paper spans the entire $(3 \times 3)$ operator space. Therefore, our proposed code can correct any single error on the codeword. This code outperforms the previous 9-qutrit code in (i) the number of qutrits required for encoding, (ii) our code can correct any arbitrary $(3 \times 3)$ error, (ii) our code can readily correct bit errors in a single step as opposed to the two-step correction used previously, and (iii) phase error correction does not require correcting individual subspaces.
\end{abstract}

\begin{IEEEkeywords}
Quantum Error Correction, Ternary Quantum Computing, Multivalued Logic.
\end{IEEEkeywords}

%
\IEEEpeerreviewmaketitle

\section{Introduction}
%
%
%
%
\IEEEPARstart{Q}{uantum} computers hold the promise of solving certain problems faster than their known classical counterparts \cite{Shor:1997:PAP:264393.264406,Grover:1996:FQM:237814.237866}. The functional unit of quantum computer, termed as qubit, is represented as a unit vector in Hilbert Space. 
Quantum systems are inherently multi-valued. Multi-valued quantum computing is gaining importance due to its ability to represent a larger search space using less resource \cite{gokhale2019asymptotic}. Increasing the search space makes cryptographic protocols more secure \cite{PhysRevLett.85.3313}. Furthermore, quantum random walk algorithms on graph often deal with higher dimensional coins \cite{Aharonov:2001:QWG:380752.380758, wong2015grover}. Ternary system is the simplest higher dimensional quantum system. An arbitrary ternary quantum state, or qutrit, has the form $\ket{\psi} = \alpha\ket{0} + \beta\ket{1} + \gamma\ket{2}$, $\alpha,\beta,\gamma \in \mathbb{C}$, $|\alpha|^2 + |\beta|^2 + |\gamma|^2 = 1$.

Evolution of a quantum state is governed by unitary operators, called quantum gates \cite{nielsen2002quantum}. The quantum system can interact with the environment and undergo some unwanted evolution ($E$), called \emph{error}. Error correction aims to undo this evolution. The first quantum error correcting code (QECC) encoded the information of a single qubit into nine qubits in order to correct a single error \cite{PhysRevA.52.R2493}. The 5-qubit QECC \cite{PhysRevLett.77.198} was shown to be optimal in the number of qubits. While the above mentioned codes are examples of concatenated code \cite{nielsen2002quantum}, a new family of code termed as surface code \cite{kitaev2003fault} has also been studied extensively, which solves the nearest-neighbour problem in Quantum Error Correction \cite{fowler2012surface,bravyi1998quantum,wang2003confinement}.

In contrast, there are fewer studies on multi-valued QECC \cite{PhysRevA.55.R839,959288,1715533}. Multi-valued QECC for erasure channel was proposed in \cite{muralidharan2017overcoming} and the channel capacity was achieved using Quantum Reed-Solomon Code \cite{PhysRevA.97.052316}. In \cite{8248788}, the authors studied error correction for multi-valued amplitude damping channel. However, for this paper, we concentrate particularly on error correction for \emph{ternary} quantum systems. A generalized error model was considered for ternary quantum systems in \cite{PhysRevA.97.052302}, and a single error was corrected using 9 qutrits. Nevertheless, this error model also could not correct arbitrary $(3\times 3)$ quantum errors. Furthermore, the 9-qutrit code required two cascading steps for correcting bit error, and phase error was corrected by correcting pairwise subspaces of the basis states. The authors did not propose the set of stabilizers for that code, and the two-step correction procedure increases the complexity, and hence the time requirement, for error correction.

In this paper, we propose a 5-qutrit quantum error correcting code. We provide the encoding for the proposed 5-qutrit QECC which is capable of correcting a single error on the codeword. We have proved that the error model considered in this paper encompasses all $(3 \times 3)$ unitary operator, and hence our proposed QECC can correct any arbitrary single error on a qutrit. We have also provided the stabilizer formulation for this QECC, which is similar to the 5-qubit code of Laflamme \cite{PhysRevLett.77.198}. Unlike the previous 9-qutrit code, our proposed set of stabilizers can readily correct both bit and phase errors in single steps. The provable lower bound on the number of qutrits is five, which makes our proposed code optimal.

The rest of the paper is organized as follows - in Section 2 we explore the error model, and show that it encompasses all possible $(3 \times 3)$ errors. In section 3 we provide the encoding scheme for the proposed QECC, and discuss its stabilizer formulation in Section 4. We conclude in Section 5.

\section{Error Model for Ternary Quantum Systems}
An error, being a quantum operator, is a unitary matrix \cite{nielsen2002quantum}. For ternary systems, any $(3 \times 3)$ unitary operator is a probable error. Consider a $(3 \times 3)$ matrix of the form
\begin{equation}
\label{eq:matrix}
    M = \begin{pmatrix}
    a & b & c\\
    d & e & f\\
    g & h & j
    \end{pmatrix}
\end{equation}

where, in general, each of the values $a, b, \hdots, j \in \mathbb{C}$. This matrix is not unitary for all values of $a, b, \hdots, j$ and hence does not represent a quantum error. Nevertheless, we stick with this matrix for the time being. Now, we consider the set of matrices $\sigma_i, i = 1, 2, \hdots 9$.

\begin{center}
	\begin{tabular}{c c}
		$\sigma_1 = \begin{pmatrix}
		1 & 0 & 0\\
		0 & 0 & 1\\
		0 & 1 & 0
		\end{pmatrix}$
		&
		$\sigma_2 = \begin{pmatrix}
		1 & 0 & 0\\
		0 & 0 & \omega^2\\
		0 & \omega & 0
		\end{pmatrix}$\\
		
		$\sigma_3 = \begin{pmatrix}
		1 & 0 & 0\\
		0 & 0 & \omega\\
		0 & \omega^2 & 0
		\end{pmatrix}$
		&
		$\sigma_4 = \begin{pmatrix}
		0 & 0 & 1\\
		0 & 1 & 0\\
		1 & 0 & 0
		\end{pmatrix}$\\

		$\sigma_5 = \begin{pmatrix}
		0 & 0 & \omega^2\\
		0 & 1 & 0\\
		\omega & 0 & 0
		\end{pmatrix}$
		&
		$\sigma_6 = \begin{pmatrix}
		0 & 0 & \omega\\
		0 & 1 & 0\\
		\omega^2 & 0 & 0
		\end{pmatrix}$\\
		
		$\sigma_7 = \begin{pmatrix}
		0 & 1 & 0\\
		1 & 0 & 0\\
		0 & 0 & 1
		\end{pmatrix}$
		&
		$\sigma_8 = \begin{pmatrix}
		0 & \omega^2 & 0\\
		\omega & 0 & 0\\
		0 & 0 & 1
		\end{pmatrix}$\\
		
		$\sigma_9 = \begin{pmatrix}
		0 & \omega & 0\\
		\omega^2 & 0 & 0\\
		0 & 0 & 1
		\end{pmatrix}$
	\end{tabular}
\end{center}

where $\omega$ is the cube-root of unity.

\begin{lemma}
\label{lemma1}
$\sigma_i$, $i = 1,2,\hdots,9$ are linearly independent.
\end{lemma}

\begin{proof}
Let us assume there exists $\Lambda_i, i = 1,2,\hdots,9$ such that $\sum \limits_{i=1}^{9}\Lambda_i\sigma_i = 0$, and $\Lambda_i \neq 0$ $\forall i$. We have the set of equations:
\begin{eqnarray*}
\Lambda_1 + \Lambda_2 + \Lambda_3 &=& 0\\
\Lambda_1 + \omega^2\Lambda_2 + \omega\Lambda_3 &=& 0\\
\Lambda_1 + \omega\Lambda_2 + \omega^2\Lambda_3 &=& 0\\
\Lambda_4 + \Lambda_5 + \Lambda_6 &=& 0\\
\Lambda_4 + \omega^2\Lambda_5 + \omega\Lambda_6 &=& 0\\
\Lambda_4 + \omega\Lambda_5 + \omega^2\Lambda_6 &=& 0\\
\Lambda_7 + \Lambda_8 + \Lambda_9 &=& 0\\
\Lambda_7 + \omega^2\Lambda_8 + \omega\Lambda_9 &=& 0\\
\Lambda_7 + \omega\Lambda_8 + \omega^2\Lambda_9 &=& 0
\end{eqnarray*}

Note that these nine equations can be grouped into three sets, each set containing three equations. No two sets of equations involve the same coefficients. In the above nine equations, the first three, second three, and the last three equations form such sets. We show the proof for the set of first three equations involving coefficients $\Lambda_1,\Lambda_2,\Lambda_3$. The proof for the other two sets will be identical.

If the set of matrices $\sigma_i$, $i = 1,2,\hdots,9$ are not linearly independent, at least two of the three coefficients $\Lambda_1,\Lambda_2,\Lambda_3$ must be non-zero. If only one of them is non-zero, then it does not satisfy the first equation, and hence such a scenario is ruled out. Without loss of generality, let us assume that $\Lambda_1 = -(\Lambda_2 + \Lambda_3) \neq 0$. Substituting $\Lambda_1$ in the second and third equations respectively yields
\begin{eqnarray*}
(\omega^2 - 1)\Lambda_2 &=& (1-\omega)\Lambda_3 \Rightarrow \frac{\Lambda_2}{\Lambda_3} = \frac{1-\omega}{\omega^2-1}\\
(\omega - 1)\Lambda_2 &=& (1-\omega^2)\Lambda_3 \Rightarrow \frac{\Lambda_2}{\Lambda_3} = \frac{1-\omega^2}{\omega-1}
\end{eqnarray*}
Equating the ratios of $\Lambda_2$ and $\Lambda_3$ gives $\omega = \omega^2$, which is not possible. Therefore, to satisfy the first set of three equations, each of the coefficients must be zero.

Extending the similar argument to the other two sets dictates that if $\sum \limits_{i=1}^{9}\Lambda_i\sigma_i = 0$, then $\Lambda_i = 0$ $\forall i$.
\end{proof}

\begin{lemma}
\label{lemma2}
For any $(3 \times 3)$ matrix $M$, $\exists \lambda_i, i = 1,2,\hdots,9$ such that $\sum\limits_{i=1}^{9}\lambda_i\sigma_i = M$.
\end{lemma}

\begin{proof}
Consider the matrix $M$ as in Eq.~\ref{eq:matrix}. Putting $M = \sum\limits_{i=1}^{9}\lambda_i\sigma_i$ yields the following 9 equations
\begin{eqnarray*}
\lambda_1 + \lambda_2 + \lambda_3 &=& a\\
\lambda_1 + \omega^2\lambda_2 + \omega\lambda_3 &=& f\\
\lambda_1 + \omega\lambda_2 + \omega^2\lambda_3 &=& h\\
\lambda_4 + \lambda_5 + \lambda_6 &=& e\\
\lambda_4 + \omega^2\lambda_5 + \omega\lambda_6 &=& c\\
\lambda_4 + \omega\lambda_5 + \omega^2\lambda_6 &=& g\\
\lambda_7 + \lambda_8 + \lambda_9 &=& i\\
\lambda_7 + \omega^2\lambda_8 + \omega\lambda_9 &=& b\\
\lambda_7 + \omega\lambda_8 + \omega^2\lambda_9 &=& d
\end{eqnarray*}
The determinant of the L.H.S. of the first three equations is always non-zero, which implies that it is always possible to find $\lambda_1, \lambda_2, \lambda_3$ which satisfy the set of three equations. Since each set of three equations has disjoint set of coefficients, similar arguments hold for the other two sets also.

Therefore, for any such matrix $M$, it is always possible to find linearly independent parameters $\lambda_i, i = 1,2,\hdots,9$ such that $\sum\limits_{i=1}^{9}\lambda_i\sigma_i = M$.
\end{proof}

\begin{theorem}
A QECC which can correct the matrices $\sigma_i, i=1,2,\hdots,9$ can correct any error on a qutrit.
\end{theorem}

\begin{proof}
The proof follows directly from Lemma~\ref{lemma1} and ~\ref{lemma2}. If $\{M\}$ is the set of all $(3 \times 3)$ matrices, then the set of all possible quantum errors $\{E\}$, such that every $\mathcal{E} \in \{E\}$ is a unitary matrix, is $\{E\} \subset \{M\}$. Therefore, any $\mathcal{E} \in \{E\}$ can also be written as a linear combination of $\sigma_i, i=1,2,\hdots,9$. If a QECC can correct each of $\sigma_i$, it can also correct any error $\mathcal{E}$ on the quantum system.
\end{proof}

If we consider the matrices $\sigma_i, i = 1,2,\hdots,9$ as errors, then it is easy to see that each of those matrices keep one of the basis states unchanged, and affects the other two bases. For example, if the quantum state is $\ket{\psi} = \alpha\ket{0} + \beta\ket{1} + \gamma\ket{2}$, then $\sigma_1\ket{\psi} = \alpha\ket{0} + \beta\ket{2} + \gamma\ket{1}$. The basis state $\ket{0}$ remains unaffected while the other two bases are flipped. Such errors prove to be difficult to correct in our framework. Hence, we aim for some further tuning. We consider the matrices $X_1$, $X_2$, $Z_1$ and $Z_2$ as errors and show that all the $\sigma_i, i = 1,2,\hdots,9$ matrices can be written as the linear combination of these matrices and their products. Therefore, any $(3 \times 3)$ quantum error $\mathcal{E}$ can be written as a linear combination of $X_1$, $X_2$, $Z_1$ and $Z_2$ and their products.

\begin{center}
	\begin{tabular}{c  c}
		$X_1 = \begin{pmatrix}
		0 & 0 & 1\\
		1 & 0 & 0\\
		0 & 1 & 0
		\end{pmatrix}$
		&
		$X_2 = \begin{pmatrix}
		0 & 1 & 0\\
		0 & 0 & 1\\
		1 & 0 & 0
		\end{pmatrix}$
	\end{tabular}
\end{center}

\begin{center}
	\begin{tabular}{ c  c}
		$Z_1 = \begin{pmatrix}
		1 & 0 & 0\\
		0 & \omega & 0\\
		0 & 0 & \omega^2
		\end{pmatrix}$
		&
		$Z_2 = \begin{pmatrix}
		1 & 0 & 0\\
		0 & \omega^2 & 0\\
		0 & 0 & \omega
		\end{pmatrix}$
	\end{tabular}
\end{center}

We explicitly show the formulation (upto a scalar coefficient) of $\sigma_1$ and $\sigma_2$ matrices using $X_1$, $X_2$, $Z_1$ and $Z_2$. The other matrices can also be formulated similarly.
\begin{eqnarray*}
\begin{pmatrix}
1 & 0 & 0\\
0 & 0 & 1\\
0 & 1 & 0
\end{pmatrix} &=& \begin{pmatrix}
1 & 0 & 0\\
0 & 1 & 0\\
0 & 0 & 1
\end{pmatrix} + \begin{pmatrix}
1 & 0 & 0\\
0 & \omega & 0\\
0 & 0 & \omega^2
\end{pmatrix}\\
&+& \begin{pmatrix}
1 & 0 & 0\\
0 & \omega^2 & 0\\
0 & 0 & \omega
\end{pmatrix} + \begin{pmatrix}
0 & 1 & 0\\
0 & 0 & 1\\
1 & 0 & 0
\end{pmatrix}\\
&+& \begin{pmatrix}
0 & \omega & 0\\
0 & 0 & 1\\
\omega^2 & 0 & 0
\end{pmatrix} + \begin{pmatrix}
0 & \omega^2 & 0\\
0 & 0 & 1\\
\omega & 0 & 0
\end{pmatrix}\\
&+& \begin{pmatrix}
0 & 0 & 1\\
1 & 0 & 0\\
0 & 1 & 0
\end{pmatrix} + \begin{pmatrix}
0 & 0 & \omega^2\\
\omega & 0 & 0\\
0 & 1 & 0
\end{pmatrix}\\
&+& \begin{pmatrix}
0 & 0 & \omega\\
\omega^2 & 0 & 0\\
0 & 1 & 0
\end{pmatrix}\\
&=& I + Z_1 + Z_2 + X_2 + \omega Z_2X_2 + \omega^2 Z_1X_2\\
& & + X_1 + \omega^2 Z_2X_1 + \omega Z_1X_1
\end{eqnarray*}
\begin{eqnarray*}
\begin{pmatrix}
1 & 0 & 0\\
0 & 0 & \omega^2\\
0 & \omega & 0
\end{pmatrix} &=& \begin{pmatrix}
1 & 0 & 0\\
0 & \omega^2 & 0\\
0 & 0 & \omega
\end{pmatrix}\begin{pmatrix}
1 & 0 & 0\\
0 & 0 & 1\\
0 & 1 & 0
\end{pmatrix}\\
&=& I + Z_1 + Z_2 + Z_2X_2 + \omega Z_1X_2 + \omega^2 X_2\\
& & + Z_2X_1 + \omega^2 Z_1X_1 +\omega X_1
\end{eqnarray*}

In accordance with \cite{PhysRevA.97.052302}, we call the errors $X_1$ and $X_2$ as ``bit shift errors", and the errors $Z_1$ and $Z_2$ as ``phase errors". The action of these errors on a general qutrit $\ket{\psi} = \alpha\ket{0} + \beta\ket{1} + \gamma\ket{2}$ are shown below:
\begin{eqnarray*}
    Z_i\ket{\psi} &=& \alpha\ket{0} + \omega^i\beta\ket{1} + \omega^{2i}\gamma\ket{2}\\
    X_i\ket{\psi} &=& \alpha\ket{0+i} + \beta\ket{1+i} + \gamma\ket{2+i}
\end{eqnarray*}
where $i \in \{1,2\}$ and the addition is modulo 3. Therefore, the error model to correct any $(3 \times 3)$ quantum error can be summarized as in Eq.~\ref{eq:model}.
\begin{equation} \label{eq:model}
\mathcal{E} = \delta\mathbb{I}_3 + \sum_{i = 1}^2\eta_iZ_i + \sum_{j=1}^{2}(\mu_jX_j + \sum_{i,j}\xi_{ij}Y_{ij})
\end{equation}
where $\mathbb{I}_3$ is the $(3 \times 3)$ identity matrix, $X_i$ and $Z_i$ are the bit and phase errors respectively, and $Y_{ij} = X_iZ_j$; $\delta,\eta,\mu,\xi \in \mathbb{C}$.

\section{Ternary Quantum Error Correction}

\subsection{5-Qutrit Error Correcting Code}
The quantum information in a single qutrit $\ket{\psi} = \alpha\ket{0} + \beta\ket{1} + \gamma\ket{2}$ is distributed into five qutrits in order to protect the information from a single error. The encoded logical qutrit is $\ket{\psi_L} = \alpha\ket{0_L} + \beta\ket{1_L} + \gamma\ket{2_L}$ where

\begin{eqnarray*}
\ket{0_L} &=& \ket{00000} + \ket{01001} + \ket{02002} + \ket{10100}\\
&+& \omega\ket{11101} + \omega^2\ket{12102} + \ket{20200} + \omega^2\ket{22211}\\
&+& \omega\ket{22202} + \ket{01010} + \omega\ket{02011} + \omega^2\ket{00012}\\
&+& \omega\ket{11110} + \ket{12111} + \omega\ket{10112} + \omega^2\ket{21210}\\
&+& \omega^2\ket{22211} + \omega^2\ket{20212} + \ket{02020} + \omega^2\ket{00021}\\
&+& \omega\ket{01022} + \omega^2\ket{12120} + \omega^2\ket{10121} + \omega^2\ket{11122}\\
&+& \omega\ket{22220} + \omega^2\ket{20221} + \ket{21222} + \ket{00101}\\
&+& \omega\ket{01102} + \omega^2\ket{02100} + \omega\ket{10201} + \ket{11202}\\
&+& \omega^2\ket{12200} + \omega^2\ket{20001} + \omega^2\ket{21002} + \omega^2\ket{22000}\\
&+& \ket{00202} + \omega^2\ket{01200} + \omega\ket{02201} + \omega^2\ket{10002}\\
&+& \omega^2\ket{11000} + \omega^2\ket{12001} + \omega\ket{20102} + \omega^2\ket{21100}\\
&+& \ket{22101} + \omega\ket{01111} + \ket{02112} + \omega^2\ket{00110}\\
&+& \ket{11211} + \ket{12212} + \ket{10210} + \omega^2\ket{21011}\\
&+& \ket{22012} + \omega\ket{20010} + \omega^2\ket{01212} + \omega^2\ket{02210}\\
&+& \omega^2\ket{00211} + \omega^2\ket{11012} + \ket{12010} + \omega\ket{10011}\\
&+& \omega^2\ket{21112} + \omega\ket{22110} + \ket{20111} + \omega^2\ket{02121}\\
&+& \omega^2\ket{00122} + \omega^2\ket{01120} + \omega^2\ket{12221} + \ket{10222}\\
&+& \omega\ket{11220} + \omega^2\ket{22021} + \omega\ket{20022} + \ket{21020}\\
&+& \omega\ket{02222} + \omega^2\ket{00220} + \ket{01221} + \omega^2\ket{12022}\\
&+& \omega\ket{10020} + \ket{11021} + \ket{22122} + \ket{20120}\\
&+& \ket{21121}
\end{eqnarray*}

The other two logical states are reported in Appendix~\ref{logical}. For error correction, the logical states should conform to the Knill-Laflamme condition \cite{PhysRevLett.84.2525} which states that for any error $E$, $\bra{0_L}E\ket{0_L} = \bra{1_L}E\ket{1_L} = \bra{2_L}E\ket{2_L}$. One can check easily that the condition is satisfied for this encoding.

\subsection{Optimality Of 5-Qutrit Code}
There are two bit errors ($X_i$), two phase errors ($Z_j$), and therefore four possible $Y_{ij} = X_i Z_j$ errors in the error model. Therefore, there are eight possible error states for each of the physical qutrits and one error free state. In an $n$-qutrit code, $(8n+1)$ error states are possible for each of the $\ket{0}_L$, $\ket{1}_L$ and $\ket{2}_L$ states, leading to $3(8n+1)$ possible error states. For successful error correction, it is necessary that these possible error states reside in orthogonal subspaces of the Hilbert Space associated with the logical qutrit. Since an $n$-qutrit state is associated with a $3^n$-dimensional Hilbert Space, the necessary condition is
\begin{center}
    $3(8n+1) \leq 3^n$
\end{center}
The above inequality is satisfied for $n \geq 5$. Therefore, a single qutrit of information must be distributed into at least $5$ qutrits in order to correct a single error. Our proposed QECC requires 5 qutrits for encoding, and is, therefore, optimal in the number of qutrits. This bound, also known as Quantum Hamming Bound, is derivable from \cite{PhysRevA.55.R839}.

\subsection{Performance Analysis Of The Code}
Our proposed QECC can correct a single error only and fails if more than one error occur on the encoded qutrit. As the set of stabilizers (shown in next section) indicates, this code is not a CSS code. Hence, this code is unable to correct a single bit and phase errors together even when they occur on different qutrits \cite{devitt2013quantum}. If $p$ is the probability that a single qutrit is erroneous, then the probability that the code fails is
\begin{eqnarray*}
    & & 1 - (1-p)^5 - \begin{pmatrix}
5\\
1
\end{pmatrix} p(1-p)^4\\
& = & 1 - (1+4p)(1-p)^4\\
& = & 10p^2 + \mathcal{O}(p^3)
\end{eqnarray*}
If no error correcting code is applied, then the probability of error on the logical qutrit $p$ is equal to the probability of error of the physical qutrit. However, if the physical qutrit is encoded using our proposed QECC, the probability of error on the logical qutrit is reduced to $10p^2 + \mathcal{O}(p^3)$. Since $p \simeq 10p^2 + \mathcal{O}(p^3)$ for $p = \frac{1}{10}$, for $p < \frac{1}{10}$, this technique provides an improved method for preserving the coherence of the qutrits.

\section{Stabilizer Formulation}
\subsection{Error Correction Via Stabilizers}
A set of operators $S_1, S_2, \hdots S_m \in \{I, \sigma_x, \sigma_z\}^{\otimes n}$ is said to stabilize a quantum state $\ket{\psi}$ if the following criteria are satisfied \cite{gottesman1997stabilizer}

\begin{enumerate}
    \item $S_i\ket{\psi} = \ket{\psi}$ $\forall$ $i$.
    \item For all errors $E$, $\exists$ $j$ such that $S_j(E\ket{\psi}) = -(E\ket{\psi})$. The -1 phase is for binary quantum systems. For ternary systems, this condition will be updated as $S_j(E\ket{\psi}) = \omega(E\ket{\psi})$ or $S_j(E\ket{\psi}) = \omega^2(E\ket{\psi})$.
    \item For different errors $E$ and $E'$, $\exists$ $j, k$ such that $S_j(E\ket{\psi}) \neq S_k(E'\ket{\psi})$.
    \item $\forall$ $i,j$, $[S_i, S_j] = 0$.
\end{enumerate}

Furthermore, an $n$-qudit state with $m$ stabilizers can encode $k = n-m$ logical qudits. Therefore, the task of error correction becomes equivalent to finding a set of stabilizers for the encoded quantum state.

The stabilizers for Shor \cite{PhysRevA.52.R2493} and Steane code \cite{PhysRevLett.77.793}, called CSS codes, can be partitioned into two disjoint sets of operators, where one set of operators ($\{S_x\}$) consists only of $I$ and $\sigma_x$, and the other set ($\{S_z\}$) only of $I$ and $\sigma_z$. The corresponding circuit of such codes are easy to implement \cite{devitt2013quantum}. Furthermore, these codes can also correct a single $\sigma_x$ and $\sigma_z$ errors if they occur on different qubits. The circuit realization of non-CSS codes (e.g. \cite{PhysRevLett.77.198}) is resource-extensive \cite{devitt2013quantum,majumdar2017method}, and such codes cannot correct two errors in any scenario \cite{devitt2013quantum}.

However, for a 5-qutrit code, it is not possible to have a CSS type stabilizer structure. From the error model, each qutrit can incur two types of bit (phase) errors. Furthermore, only a single error is assumed to have occurred on the codeword. This accounts for $2 \times 5 = 10$ possible single bit (phase) error combinations. Each of the stabilizers can have one of the three possible outcomes: $1, \omega, \omega^2$. Hence at least 3 stabilizers are required to detect 10 bit (phase) error combinations and the error free state. Therefore, in order to obtain a set of operators $S = \{S_x,S_z\}$, $|S| \geq 6$. However, from the equation $k = n - m$, it is evident that for a 5-qutrit QECC, there are four stabilizers which encode a single logical qutrit. Therefore, the 5-qutrit code is a non-CSS code.


\subsection{Stabilizer Structure For The 5-Qutrit QECC}
The general stabilizer structure for higher dimensional quantum systems, as proposed by Gottesman \cite{gottesman1998fault}, is

\begin{center}
    $X\ket{j} = \ket{j+1}$ mod $d$ \hspace*{0.7cm} $Z\ket{j} = \omega^j\ket{j}$
\end{center}

where $d$ is the dimension of the quantum system.

The proposed stabilizers for correcting bit errors in the codeword are as follows:
\begin{eqnarray*}
S_1 & = & I \otimes X \otimes Z \otimes Z \otimes X\\
S_2 & = & X \otimes I \otimes X \otimes Z \otimes Z\\
S_3 & = & Z \otimes X \otimes I \otimes X \otimes Z\\
S_4 & = & Z \otimes Z \otimes X \otimes I \otimes X
\end{eqnarray*}

These set of stabilizers is equivalent to the 5-qubit code \cite{PhysRevLett.77.198}. However, one can come up with a different set of four stabilizers which leads to different codewords. In Table~\ref{tab:stab} we show the action of the stabilizers on the errors $X_1,X_2,Z_1,Z_2$ when the error occurs on the first qutrit only. The eigenvalues corresponding to the stabilizers for errors occurring on other qutrits can be obtained similarly.

\begin{table}[htb]
    \centering
    \caption{Error Correction with Stabilizers for Errors on the First Qutrit}
    \begin{tabular}{|c|c|c|c|c|}
    \hline
        Error State & $S_1$ & $S_2$ & $S_3$ & $S_4$ \\
        \hline
        $\ket{\psi}$ & +1 & +1 & +1 & +1\\
        \hline
        $X_1\ket{\psi}$ & +1 & +1 & $\omega$ & $\omega$\\
        $X_2\ket{\psi}$ & +1 & +1 & $\omega^2$ & $\omega^2$\\
        $Z_1\ket{\psi}$ & +1 & $\omega^2$ & +1 & +1\\
        $Z_2\ket{\psi}$ & +1 & $\omega$ & +1 & +1\\
        \hline
    \end{tabular}
    \label{tab:stab}
\end{table}

A major shortcoming of the 9-qutrit code \cite{PhysRevA.97.052302} is that it fails to provide a complete set of stabilizers for error correction. Therefore, the bit error correction is two step - first to detect the presence of error, and then to determine its location. Similarly, for phase errors, individual subspaces were corrected. Since there are three subspaces ($\{\ket{0},\ket{1}\}$,$\{\ket{1},\ket{2}\}$,$\{\ket{2},\ket{0}\}$), this is a three step process. Our proposed QECC overcomes these shortcomings. Operating the proposed four stabilizers on the codeword can correct a single error. Therefore, two (three) steps are not required for the correction of a single bit (phase) error.

The circuit realization of Laflamme's 5-qubit code is not trivial. Furthermore, the gates used in the circuit of that code are not mostly implementable in modern day technology. It was shown in \cite{majumdar2017method} that the quantum cost of that circuit is significantly higher than that of Shor and Steane code. We invite the readers to come up with a circuit for this 5-qutrit code. To the best of our knowledge, no universal gate set is available for ternary quantum systems. However, the MS gates \cite{PhysRevA.62.052309} are shown to be implementable in Ion-Trap technology. It may be worthwhile to try to realize the circuit for this proposed code using MS gates and Chrestenson gates \cite{Hurst1985-HURSTI-2}.

\section{Conclusion}
In this paper we have proposed a 5-qutrit error correcting code which can correct any arbitrary $(3 \times 3)$ error on a qutrit. The error model considered in this paper overcomes the shortcomings of specialized channels considered in the previous literature for ternary QECC. Five qutrits are necessary for correcting a single error in a qutrit. Hence, our code is optimal in the number of qutrits. The stabilizer structure of this code is similar to the 5-qubit QECC. Furthermore, our proposed code can correct both bit and phase error in single steps, as compared to multi-step corrections in the previous 9-qutrit code. The future scope is to come up with the circuit for this code, and also to search for the minimum qutrit CSS code.

\section*{Acknowledgement}
Ritajit Majumdar would like to acknowledge Prof. Debasis Sarkar, University of Calcutta, for helpful discussions on the error model.


%

\appendix

\section{Encoding scheme for the qutrits}
\label{logical}

The encoded logical qutrit is $\ket{\psi_L} = \alpha\ket{0_L} + \beta\ket{1_L} + \gamma\ket{2_L}$ where

\begin{eqnarray*}
\ket{0_L} &=& \ket{00000} + \ket{01001} + \ket{02002} + \ket{10100}\\
&+& \omega\ket{11101} + \omega^2\ket{12102} + \ket{20200} + \omega^2\ket{22211}\\
&+& \omega\ket{22202} + \ket{01010} + \omega\ket{02011} + \omega^2\ket{00012}\\
&+& \omega\ket{11110} + \ket{12111} + \omega\ket{10112} + \omega^2\ket{21210}\\
&+& \omega^2\ket{22211} + \omega^2\ket{20212} + \ket{02020} + \omega^2\ket{00021}\\
&+& \omega\ket{01022} + \omega^2\ket{12120} + \omega^2\ket{10121} + \omega^2\ket{11122}\\
&+& \omega\ket{22220} + \omega^2\ket{20221} + \ket{21222} + \ket{00101}\\
&+& \omega\ket{01102} + \omega^2\ket{02100} + \omega\ket{10201} + \ket{11202}\\
&+& \omega^2\ket{12200} + \omega^2\ket{20001} + \omega^2\ket{21002} + \omega^2\ket{22000}\\
&+& \ket{00202} + \omega^2\ket{01200} + \omega\ket{02201} + \omega^2\ket{10002}\\
&+& \omega^2\ket{11000} + \omega^2\ket{12001} + \omega\ket{20102} + \omega^2\ket{21100}\\
&+& \ket{22101} + \omega\ket{01111} + \ket{02112} + \omega^2\ket{00110}\\
&+& \ket{11211} + \ket{12212} + \ket{10210} + \omega^2\ket{21011}\\
&+& \ket{22012} + \omega\ket{20010} + \omega^2\ket{01212} + \omega^2\ket{02210}\\
&+& \omega^2\ket{00211} + \omega^2\ket{11012} + \ket{12010} + \omega\ket{10011}\\
&+& \omega^2\ket{21112} + \omega\ket{22110} + \ket{20111} + \omega^2\ket{02121}\\
&+& \omega^2\ket{00122} + \omega^2\ket{01120} + \omega^2\ket{12221} + \ket{10222}\\
&+& \omega\ket{11220} + \omega^2\ket{22021} + \omega\ket{20022} + \ket{21020}\\
&+& \omega\ket{02222} + \omega^2\ket{00220} + \ket{01221} + \omega^2\ket{12022}\\
&+& \omega\ket{10020} + \ket{11021} + \ket{22122} + \ket{20120}\\
&+& \ket{21121}
\end{eqnarray*}

\begin{eqnarray*}
\ket{1_L} &=& \ket{11111} + \omega^2\ket{12112} + \omega\ket{10110} + \omega^2\ket{21211}\\
&+& \omega^2\ket{22212} + \omega^2\ket{20210} + \omega\ket{01011} + \omega^2\ket{02012}\\
&+& \ket{00010} + \omega^2\ket{12121} + \omega^2\ket{10122} + \omega^2\ket{11120}\\
&+& \omega^2\ket{22221} + \ket{20222} + \omega\ket{21220} + \omega^2\ket{02021}\\
&+& \omega\ket{00022} + \ket{01020} + \omega\ket{10101} + \omega^2\ket{11102}\\
&+& \ket{12100} + \omega^2\ket{20201} + \omega\ket{21202} + \ket{22200}\\
&+& \ket{00001} + \ket{01002} + \ket{02000} + \omega^2\ket{11212}\\
&+& \omega^2\ket{12210} + \omega^2\ket{10211} + \omega^2\ket{21012} + \ket{22010}\\
&+& \omega\ket{20011} + \omega^2\ket{01112} + \omega\ket{02110} + \ket{00111}\\
&+& \omega\ket{11010} + \omega^2\ket{12011} + \ket{10012} + \omega^2\ket{21110}\\
&+& \omega\ket{22111} + \ket{20112} + \ket{01210} + \ket{02211}\\
&+& \ket{00212} + \omega^2\ket{12222} + \ket{10220} + \omega\ket{11221}\\
&+& \ket{22022} + \omega^2\ket{20020} + \omega\ket{21021} + \omega\ket{02122}\\
&+& \omega\ket{00120} + \omega\ket{01121} + \omega^2\ket{12020} + \omega\ket{10021}\\
&+& \ket{11022} + \omega\ket{22120} + \omega\ket{20121} + \ket{21122}\\
&+& \ket{02220} + \omega\ket{00221} + \omega^2\ket{01222} + \omega^2\ket{10202}\\
&+& \omega\ket{11200} + \ket{12201} + \omega\ket{20002} + \omega\ket{21000}\\
&+& \omega\ket{22001} + \ket{00102} + \omega\ket{01100} + \omega^2\ket{02101}\\
&+& \ket{10000} + \ket{11001} + \ket{12002} + \ket{20100}\\
&+& \omega\ket{21101} + \omega^2\ket{22102} + \ket{00200} + \omega^2\ket{01201}\\
&+& \omega\ket{02202}
\end{eqnarray*}

\begin{eqnarray*}
\ket{2_L} &=& \ket{22222} + \omega\ket{20220} + \omega^2\ket{21221} + \omega\ket{02022}\\
&+& \ket{00020} + \omega^2\ket{01021} + \omega^2\ket{12122} + \omega^2\ket{10120}\\
&+& \omega^2\ket{11121} + \omega\ket{20202} + \ket{21200} + \omega^2\ket{22201}\\
&+& \ket{00002} + \ket{01000} + \ket{02001} + \omega^2\ket{10102}\\
&+& \ket{11100} + \omega\ket{12101} + \omega^2\ket{21212} + \omega^2\ket{22210}\\
&+& \omega^2\ket{20211} + \omega^2\ket{01012} + \ket{02010} + \omega\ket{00011}\\
&+& \omega^2\ket{11112} + \omega\ket{12110} + \ket{10111} + \omega\ket{22020}\\
&+& \ket{20021} + \omega^2\ket{21022} + \ket{02120} + \ket{00121}\\
&+& \ket{01122} + \omega^2\ket{12220} + \ket{10221} + \omega\ket{11222}\\
&+& \omega^2\ket{22121} + \omega^2\ket{20122} + \omega^2\ket{21120} + \omega^2\ket{02221}\\
&+& \ket{00222} + \omega\ket{01220} + \omega^2\ket{12021} + \omega\ket{10022}\\
&+& \ket{11020} + \ket{20000} + \ket{21001} + \ket{22002}\\
&+& \ket{00100} + \omega\ket{01101} + \omega^2\ket{02102} + \ket{10200}\\
&+& \omega^2\ket{11201} + \omega\ket{12202} + \omega^2\ket{20101} + \ket{21102}\\
&+& \omega\ket{22100} + \ket{00201} + \omega^2\ket{01202} + \omega\ket{02200}\\
&+& \omega\ket{10001} + \omega\ket{11002} + \omega\ket{12000} + \omega^2\ket{21010}\\
&+& \ket{22011} + \omega\ket{20012} + \ket{01110} + \omega^2\ket{02111}\\
&+& \omega\ket{00112} + \omega\ket{11210} + \omega\ket{12211} + \omega\ket{10212}\\
&+& \omega^2\ket{21111} + \omega\ket{22112} + \ket{20110} + \omega\ket{01211}\\
&+& \omega\ket{02212} + \omega\ket{00210} + \ket{11011} + \omega\ket{12012}\\
&+& \omega^2\ket{10010}
\end{eqnarray*}

\ifCLASSOPTIONcaptionsoff
  \newpage
\fi



\bibliographystyle{IEEEtran}
\bibliography{bare_jrnl}
%

%








\end{document}